\documentclass[journal,11pt,letter]{IEEEtran}
\IEEEoverridecommandlockouts

\pdfoutput=1

\usepackage{cite}
\usepackage{amsmath,amssymb,amsfonts}
\usepackage{algpseudocode}
\usepackage{graphicx}
\usepackage{textcomp}
\usepackage{xcolor}
\usepackage{xfrac}
\usepackage{hyperref}
\usepackage[utf8]{inputenc} 
\usepackage[T1]{fontenc}
\usepackage{url}
\usepackage{ifthen}
\usepackage{cite}
\usepackage{graphicx}
\usepackage{amsfonts}
\usepackage{amssymb}
\usepackage{amsthm}
\usepackage{bbm}
\usepackage{subfigure}
\usepackage{diagbox}
\usepackage{tikz}
\usepackage{booktabs,multirow}
\usepackage{array}
\usepackage{caption}

\usepackage{authblk}
\usepackage[font=small]{caption}
\usetikzlibrary{decorations.markings}

\usepackage[margin=1.0in]{geometry}

\usepackage{prettyref,xspace}

\newtheorem{theorem}{Theorem}

\newcommand\encircle[1]{%
  \tikz[baseline=(X.base)] 
    \node (X) [draw, shape=circle, inner sep=0] {\strut #1};}
\tikzset{->-/.style={decoration={
  markings,
  mark=at position #1 with {\pgftransformscale{1.5}\arrow{>}}},postaction={decorate}}}
 
 \tikzset{-<-/.style={decoration={
  markings,
  mark=at position #1 with {\pgftransformscale{1.5}\arrow{<}}},postaction={decorate}}}

\def\BibTeX{{\rm B\kern-.05em{\sc i\kern-.025em b}\kern-.08em
    T\kern-.1667em\lower.7ex\hbox{E}\kern-.125emX}}

\newcommand{\figref}[1]{Fig.~\ref{#1}}  

\DeclareMathOperator{\E}{E}

\def\E{\mathbb{E}}

\def\Pr{\mathbb{P}}

\begin{document}

\onecolumn

\title{Compressed Error HARQ: Feedback Communication on Noise-Asymmetric Channels\\
\thanks{Code can be found at: \href{https://github.com/sravan-ankireddy/nams}{https://github.com/sravan-ankireddy/ce-harq}}
}

\makeatletter \renewcommand\AB@affilsepx{\hfill \protect\Affilfont} \makeatother

\author[1]{Sravan Kumar Ankireddy \thanks{Correspondence to: Sravan <sravan.ankireddy@utexas.edu>}}
\author[2]{S. Ashwin Hebbar }
\author[3]{Yihan Jiang}
\author[1]{Hyeji Kim}
\author[2]{Pramod Viswanath}
\affil[1]{University of Texas at Austin} 
\affil[2]{Princeton University} 
\affil[3]{Aira Technologies}
\renewcommand\Authands{, }

\maketitle

\begin{abstract}
In modern communication systems with feedback, there are increasingly more scenarios where the transmitter has much less power than the receiver (e.g., medical implant devices), which we refer to as {\em noise-asymmetric channels}. For such channels, the feedback link is of higher quality than the forward link. However, feedback schemes for cellular communications, such as hybrid ARQ, do not fully utilize the high-quality feedback link. 
To this end, we introduce {\em Compressed Error Hybrid ARQ}, a generalization of hybrid ARQ tailored for noise-asymmetric channels;  the receiver sends its estimated message to the transmitter, and the transmitter harmoniously switches between hybrid ARQ and compressed error retransmission. 
We show that our proposed method significantly improves reliability, latency, and spectral efficiency compared to the conventional hybrid ARQ in various practical scenarios where the transmitter is resource-constrained. 

\end{abstract}


\section{Introduction}\label{sec:intro}

With the advance of low-powered devices, there are increasingly more communication scenarios where the transmitter has much less power than the receiver. 
%
For example, low-powered \emph{medical implant devices} send periodic biological measurements to a power-supplied machine; \emph{Internet of Things (IoT) sensors} with limited battery capacity send various information to a high-powered device, e.g., Alexa; \emph{Deep Space Satellite} uses limited solar energy to send data to the receiving station on earth. In all these examples, the transmitter is power-limited, while the receiver operates at a much higher power level~\cite{heydon2012bluetooth}. 
We let \emph{noise-asymmetric channels} denote an extreme class of such channels, where the forward channel is noisy, and the feedback channel is noise-free as depicted in Fig \ref{fig:fwd_fb_model}.

\begin{figure}[ht]
    \centering
 	\includegraphics[width=0.7\linewidth]{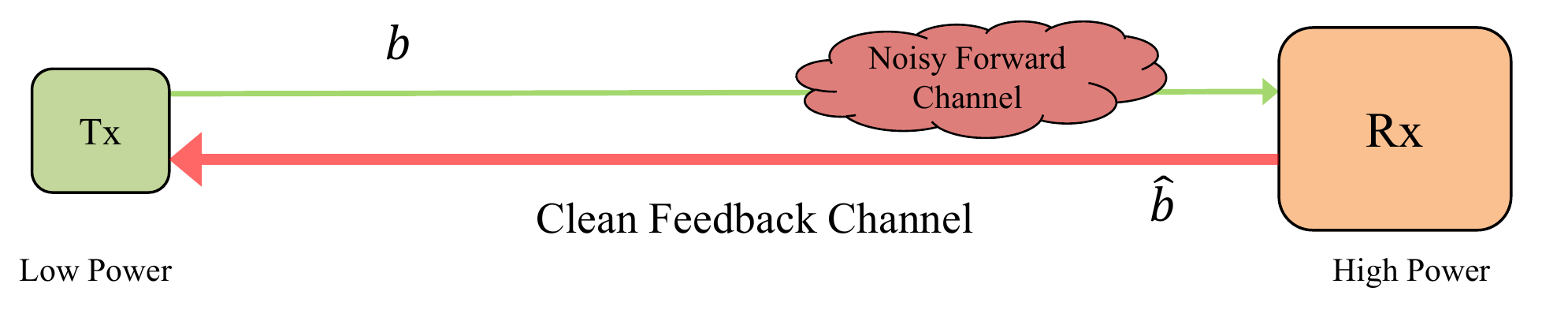}
 	\captionsetup{font=small}
 	\caption{Noise-Asymmetric Channels}
 	\label{fig:fwd_fb_model}
\end{figure}

To leverage the asymmetry in the channel, it is crucial to utilize the clean feedback channel from the receiver to the transmitter. 
%
%
%
%
Feedback schemes used in cellular communications, such as Automatic repeat request (ARQ) and Hybrid ARQ (HARQ), however, does not fully utilize the clean feedback channel. %
%
In the (Hybrid) ARQ, the receiver sends a single-bit feedback, an ACKnowledgment (ACK) or a Negative ACK (NACK), based on the Cyclic Redundancy Checks (CRC). 
The transmitter retransmits the entire coded message block (or a pre-arranged coded block) if a NACK is received. 

We focus on designing a {\em feedback communication system} for {\em noise-asymmetric channels}, which successfully leverages the high-quality feedback channels. 
Our method generalizes the HARQ based on the following insight: when the receiver's estimate has few errors, it is beneficial to transmit the error vector, which is sparse and compressible, instead of retransmitting the entire message. 
 The framework of 
 error transmission is first introduced by Ooi and Wornell in~\cite{ooi1998fast} and is recently explored in~\cite{perotti2021accumulative}. We combine this idea and HARQ to build a feedback communication algorithm for noise-asymmetric channels. (An extensive summary of existing feedback schemes~\cite{schalkwijk1966coding, chance2011concatenated, ben2017interactive, ooi1998fast, perotti2021accumulative,kim2018deepcode, safavi2021deep,  mashhadi2021drf, shao2022attentioncode} and their comparison with our work are provided in Section~\ref{sec:related}.)

Our main contributions are summarized as follows. \vspace{.2em}

\begin{itemize}
    \item \noindent{\em Algorithm design:\ } We propose the Compressed Error HARQ (CE-HARQ), a novel feedback communication algorithm for noise-asymmetric channels with block-level feedback, 
to facilitate low-latency communications in low-powered devices. Our scheme seamlessly switches between the HARQ and error compression based on analytical criteria.
\vspace{.2em}

    \item \vspace{.2em}
\noindent{\em Numerical evaluations:\ } 
We achieve significant gains over the conventional HARQ scheme on a wide range of practical resource-constrained scenarios. For low-complexity coding schemes such as convolutional codes, CE-HARQ achieves up to $2.5$ dB gain over conventional HARQ. Additionally, we also demonstrate improvements in latency and spectral efficiency.
\vspace{.2em}

    \item \vspace{.2em}
\noindent{\em Practical considerations: } 
Our design of CE-HARQ takes into consideration various resource constraints, such as systems where the redesign of the PHY layer is prohibitive. This is done by providing the flexibility to separate the channel coding into PHY and MAC layers. As a result, our algorithm is compatible with and can be smoothly applied to low-resource applications like IoT and Bluetooth communications. Additionally, we add fallback mechanisms and guidelines for the practical deployment of such systems.
\end{itemize}

\section{System Model}\label{sec:sysmodel}

We consider a block-feedback system with an Additive White Gaussian Noise (AWGN) forward channel and a noiseless feedback channel, with active feedback and BPSK modulation. The aim is to communicate a message vector $\mathbf{u} \in \{0,1\}^K$ over a maximum of $D$ rounds.

At round $i$, the encoder $\phi^{(i)}$ takes as input the message $\mathbf{u}$ and the feedback from previous round $\hat{\mathbf{u}}^{(i-1)}$, which is the message estimate, and produces $N$ symbols to be transmitted.
\begin{align}
    \mathbf{x}^{(i)} = \phi^{(i)}(\mathbf{u},\hat{\mathbf{u}}^{(i-1)}),
\end{align}
where $\mathbf{x}^{(i)} \in \{0,1\}^N$. 
The transmitter then sends the encoded vector over the forward channel as 
\begin{align}
    \mathbf{y}^{(i)} = \mathbf{x}^{(i)} + \mathbf{n}^{(i)},
\end{align}
where $\mathbf{n}^{(i)} \sim \mathcal{N}(0, \sigma^2 I) \in \mathbb{R}^N
$ is the noise vector and $\mathbf{y}^{(i)}\in \mathbb{R}^N$ is the received vector.

At the receiver, the decoder $\psi^{(i)}$ takes as input the received vector $\mathbf{y}^{(i)}$ and the previous message estimate $\hat{\mathbf{u}}^{(i-1)}$, to produce the new estimate $\mathbf{u}^{(i)}$ given by
\begin{align}
    \hat{\mathbf{u}}^{(i)} = \psi^{(i)}(\mathbf{y}^{(i)},\hat{\mathbf{u}}^{(i-1)}),
\end{align}
where $\mathbf{u}^{(i)} \in \{0,1\}^K$. The receiver then sends the new message estimate $\mathbf{u}^{(i)}$ to the transmitter, which is used for generating the next transmission $\mathbf{x}^{(i+1)}$. We reiterate that the feedback sequence $\mathbf{u}^{(i)} \in \{0,1\}^N$ is restricted to a binary sequence, making it  practical for resource-constrained settings and comparable to traditional HARQ schemes.  

This process is continued until an error-free decoding or a maximum number of rounds $D$ is reached. The objective is to design a pair of encoder-decoder $\{\phi^{(i)},\psi^{(i)}\}_{i=1}^D$ that minimizes the probability of error $\Pr\{\mathbf{u} \neq \hat{\mathbf{u}}^{(D)}\}$ for a given maximum number of rounds $D$.

\section{Compressed Error HARQ}\label{sec:algo}

We introduce the Compressed Error HARQ (CE-HARQ) scheme, which generalizes the HARQ to leverage the feedback of the {\em estimated message} from the receiver to the transmitter. 
The key innovation of CE-HARQ is achieving a very low Forward Error Correction (FEC) rate for retransmissions by compressing the {difference} between the message and the estimated message, i.e., {\em error}, as opposed to the {\em message itself}. 

This is done by selectively and carefully utilizing {\em sparsity} in error vectors. In the early rounds of communications, where the error vector is not sparse, the CE-HARQ performs the conventional HARQ. In the later rounds of communications, as the error vector becomes more sparse, it iteratively updates and compresses the error vector after each round, reducing the FEC rate further.

The CE-HARQ system has three major operational blocks, as depicted in Fig~\ref{fig:sysmodel}. 
The transmitter first determines whether to send the message or the error based on the sparsity of the error, shown in \encircle{A}. The encoder then performs two layers of channel encoding, MAC and PHY, shown in \encircle{B} and transmits through the forward channel. At the receiver, the decoder estimates the message, shown in \encircle{C}, and transmits this estimate through the high-quality feedback link.

\begin{figure}[!htb]
    \centering
 	\includegraphics[width=1.0\linewidth]{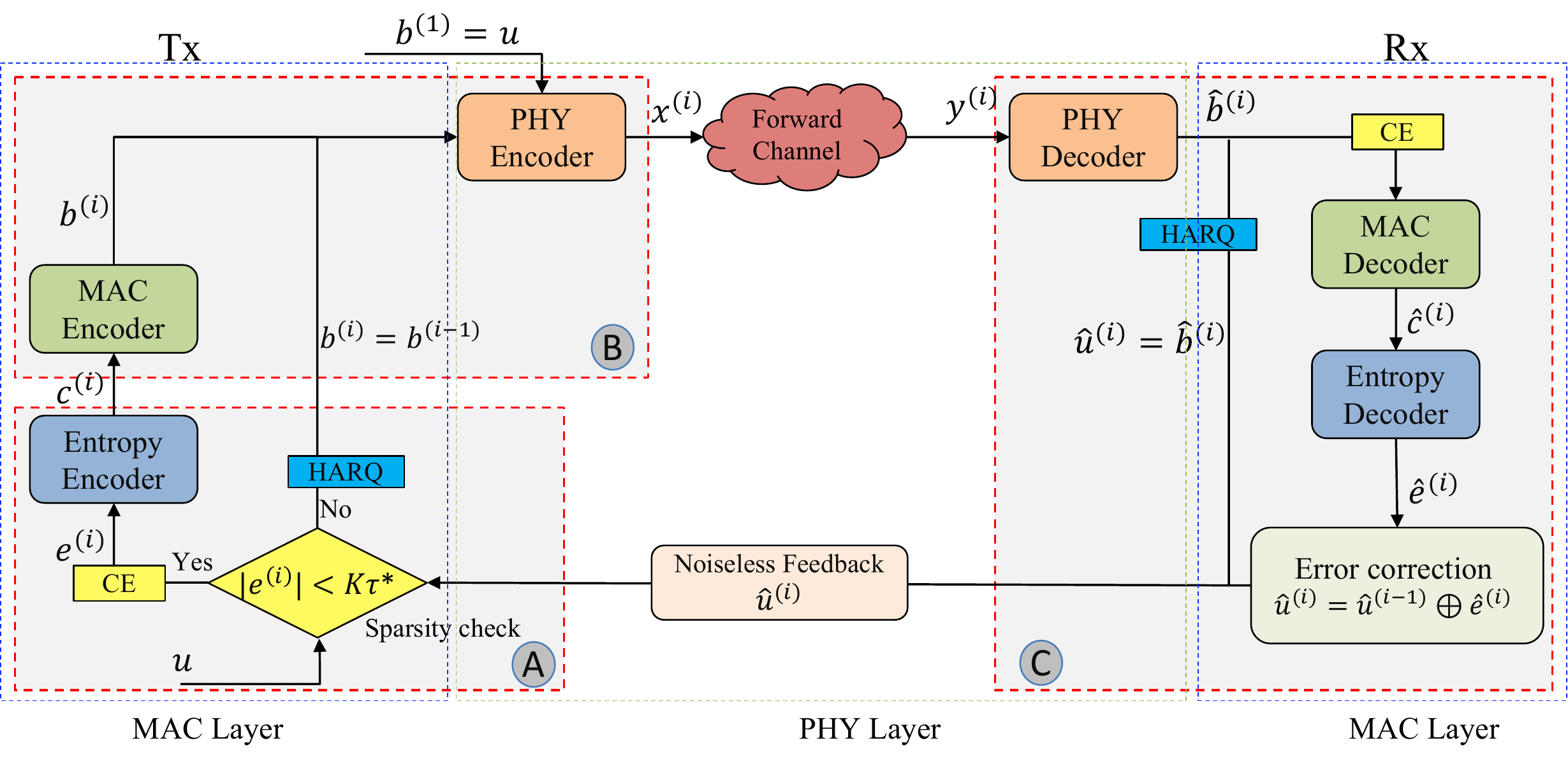}
 	\captionsetup{font=small}
 	\caption{CE-HARQ System Model}
 	\label{fig:sysmodel}
\end{figure}

\begin{table*}
\centering
\begin{tabular}{l cccc cc}
\toprule
\multirow{2}{*}{Method} &
      Feedback &   Entropy & MAC  & PHY  & Forward  & Feedback\\ 
 &     Code  &   Code  &   Code   & Code & Message  & Message    \\ 

   \midrule
    ARQ                         &    ACK/NACK   & N/A     & N/A  &   Any  & Binary & Binary                 \\
    HARQ                        &   ACK/NACK    & N/A     & HARQ &   Any     & Binary & Binary              \\
    DeepCode \cite{kim2018deepcode}                   &  Neural       & Neural  & Neural & Neural & Continuous & Continuous \\
    
    GBAF \cite{ozfatura2022all} &  Neural-Block & Neural  & Neural & Neural & Continuous & Continuous \\
    
    CEC\cite{ooi1998fast}  &  Received vector $\hat{\mathbf{x}}^{(i)}$   & Shannon-Fano  &  N/A & N/A & Binary & Binary \\
    
    AIC\cite{perotti2021accumulative}  &  Quantized LLR  & Huffman  &  N/A & N/A & Binary & Binary \\
    
    Ours (CE-HARQ)    &  Message estimate $\hat{\mathbf{u}}^{(i)}$ & AC & Configurable+HARQ & Any    & Binary & Binary \\

    \bottomrule

\end{tabular}
\caption{Comparison of various feedback schemes}
\label{tab:comparison}
\end{table*}

\subsection{HARQ vs. CE Selection and Error Compression}
The first round is always HARQ and the PHY encoder takes the message $\mathbf{u} \in \{0,1\}^K$ as input \textit{i.e.,} $\mathbf{b}^{(1)}=\mathbf{u}$. Subsequently, in the $(i+1)^{th}$ round, the transmitter first computes the error in the receiver's estimated message using the feedback from round $i$, $\hat{\mathbf{u}}^{(i)}$, as $\mathbf{e}^{(i)} = \mathbf{u} \oplus \hat{\mathbf{u}}^{(i)} \in \{0,1\}^K$, which is used to select the retransmission scheme.


\emph{Sparsity Check: HARQ or Compressed Error (CE). } 
If the number of errors in $\mathbf{e}^{(i)}$ 
is large, 
\textit{i.e.,} ${|e^{(i)}|} > {K}\tau^\star$, where $\tau^\star$ is the pre-defined sparsity threshold (discussed in Section~\ref{sec:selection}), 
we retransmit the message : $\mathbf{b}^{(i)} = \mathbf{u}$. 
On the other hand, once the error vector is highly sparse, \textit{i.e.} ${|e^{(i)}|} < {K}\tau^\star$, the scheme switches from  HARQ retransmission to compressed error retransmission \textit{i.e., CE}, detailed below. \vspace{.1em}

\emph{Entropy Encoder.} The error vector $\mathbf{e}^{(i)}$ is first losslessly compressed to $ H(\mathbf{e}^{(i)})$ bits , $\mathbf{c}^{(i)}$, 
using Arithmetic Coding (AC) \cite{cover1999elements}, which allows the use of a very low rate FEC code to protect $\mathbf{c}^{(i)}$ as shown in the next section. 
\begin{align*}
     \phi_{\text{AC}} : \mathbf{e}^{(i)} \in \{0,1\}^K \rightarrow \mathbf{c}^{(i)} \in \{0,1\}^{H(\mathbf{e}^{(i)})}.
\end{align*}

\subsection{Channel Coding}
Since the sparsity of the error vector varies for each round of CE, the achievable rate for the PHY layer FEC scheme changes with each round. In low power constrained systems such as Bluetooth Low Energy/Bluetooth Classic, strong latency requirements makes it challenging to modify the PHY layer. To address this, we use two layers of channel coding:  \vspace{.1em}

{\em MAC encoder.} We use a variable rate MAC encoder to first encode $\mathbf{c}^{(i)} \in \{0,1\}^{H(\mathbf{e}^{(i)})}$ to $K$ bits, resulting in constant rate for PHY encoder.
\begin{align*}
    \phi^{(i)}_{\text{MAC}} : \mathbf{c}^{(i)} \in \{0,1\}^{H(\mathbf{e}^{(i)})} \rightarrow \mathbf{b}^{(i)} \in \{0,1\}^K.
\end{align*}

{\em PHY encoder.} The PHY encoder is a fixed rate encoder that maps a given input $K$ bits to $N$ bits using a rate $R=\sfrac{K}{N}$ code. 
 \begin{align}
    \phi^{(i)}_{\text{PHY}} : \mathbf{b}^{(i)} \in \{0,1\}^K \rightarrow \mathbf{x}^{(i)} \in \{0,1\}^N.
\end{align}
 This two-phase MAC and PHY encoding results in an effective coding rate of $\frac{H(\mathbf{e}^{(i)})}{K} \times \frac{K}{N} = \frac{H(\mathbf{e}^{(i)})}{N}$, while keeping the PHY layer unchanged. Notably, during the HARQ phase the MAC encoding is not performed, allowing a rate of $\sfrac{K}{N}$.\vspace{.2em}

{\em Remark:} Note that we perform a separate source-channel coding for the error vector. In Appendix~\ref{sec:src-ch separation}, we show that such separation is asymptotically optimal for each round of forward communication with the receiver's message estimate available at both the transmitter and receiver.

\subsection{Decoding}
The decoding operation is performed in two stages. First, the PHY decoder estimates the $K$ bits from PHY encoder $\hat{\mathbf{b}}^{(i)}$. During HARQ scheme, this is sufficient to recover the message as $\hat{\mathbf{u}}^{(i)}$. On the other hand, for CE scheme, the MAC decoder performs another layer of decoding to estimate the compressed error $\hat{\mathbf{c}}^{(i)}$. Finally, using $\hat{\mathbf{c}}^{(i)}$, the entropy decoder reconstructs the full error vector $\hat{\mathbf{e}}^{(i)}$, which is used to get the new message estimate as $\hat{\mathbf{u}}^{(i)} = \hat{\mathbf{u}}^{(i-1)} \oplus \hat{\mathbf{e}}^{(i)}$.

\subsection{Adaptive and Iterative compression}\label{sec:selection}
At every retransmission, the transmitter chooses the retransmission scheme based on the sparsity of error vector $\mathbf{e}^{(i)}$. It is beneficial to choose the CE scheme only when the error is \textit{sufficiently sparse} so that the resulting coding gain from the low-rate FEC is higher than the loss of SNR gain by not choosing HARQ. There exists an optimal sparsity threshold $\tau^\star$, which varies dynamically with the effective SNR.

To determine $\tau^\star$ for a given effective SNR and FEC rate, we first derive an analytical solution by assuming an asymptotic setting (Appendix~\ref{sec:spa_analytical}). However, the optimality of both source and channel coding do not hold well at short block lengths. To account for this, we find a numerical solution for the optimal threshold for a given effective SNR and rate using Monte-Carlo simulations (Appendix~\ref{sec:grid_search}). Based on the $\tau^\star$ computed, the retransmission scheme is selected for each round for different SNRs. We refer to this as \textit{adaptive compression}. This is a key distinction between CE-HARQ and other feedback schemes, such as AIC, which continuously compress after each round.

During the CE phase of CE-HARQ, the error in the message estimate should decrease over the rounds, making the error vector sparser and allowing for further compression. This concept is the basis of \textit{iterative compression} in CEC~\cite{ooi1998fast} and AIC~\cite{perotti2021accumulative}. However, if an error occurs in MAC decoding,
the error may increase \textit{i.e.,} $|\mathbf{e}^{(i+1)}| > |\mathbf{e}^{(i)}|$. In such cases, CE-HARQ triggers a fallback mechanism and retransmits the previous error vector to improve the error estimate at the receiver. This guarantees that the iterative compression is performed only when the new error vector is sparser than the previous error vector, which is crucial for practical systems.

\section{Related work and comparison}\label{sec:related}

In this section, we compare CE-HARQ to existing communication systems with feedback, as summarized in Table \ref{tab:comparison}.
The most commonly used ARQ and HARQ schemes are not optimized for noise-asymmetric channels. Upon the availability of noiseless feedback, Perotti et al. introduced an alternate approach, Accumulative Iterative Code (AIC)~\cite{perotti2021accumulative}, building on the pioneering work of Ooi and Wornell - Compressed Error Cancellation (CEC)~\cite{ooi1998fast}. These schemes use feedback to communicate compressed errors without channel coding during retransmission. Notably, these schemes achieve improved spectral efficiency at the cost of the number of retransmissions, rendering them unsuitable for many latency-critical applications. In contrast, CE-HARQ augments the compressed error framework with channel coding, to maximize the probability of decoding in each round. This is achieved by utilizing the maximum available channel uses in each round, thus trading off spectral efficiency for the number of retransmissions. Additionally, AIC~\cite{perotti2021accumulative} uses Huffman coding to compress the error, which requires a large lookup table to be stored for each round of communication adding significant overhead. Instead, CE-HARQ uses Arithmetic Coding (AC) or even index-coding \cite{jiang2022indexing} when the number of errors is very small, suffering minimal overhead.


    
    



More recently, several works focused on deep learning-based feedback codes, at both symbol and block level, \cite{kim2018deepcode, safavi2021deep, mashhadi2021drf, shao2022attentioncode, ozfatura2022all} which leverage feedback to significantly improve reliability at practical block lengths. However, these systems are challenging to deploy practically since they incur significant costs in terms of memory, and complexity, and cannot be deployed in resource-constrained applications that require low latency. Further, they rely on non-binary continuous feedback rendering them unsuitable for implementation in the modern communication stack. Meanwhile, CE-HARQ uses binary feedback in conjunction with a low-complexity channel coding and error compression to ensure compatibility with current systems and low-resource settings. Recent commercial feedback communication systems, such as those from Aira Technology~\cite{jiang2022indexing,chandrasekher2022multi} are good examples of practical error compression feedback systems, that have been shown to achieve noticeable gains in reliability and range of Bluetooth-based IoT systems.




\section{Results}
In this study, we focus on resource-constrained applications on noise-asymmetric channels. Often in such applications, it is infeasible to modify the PHY layer design; CE-HARQ can support this by limiting changes to the MAC layer.

Specifically, our main result is for low-power Bluetooth communication, where the transmitter is highly power-constrained \cite{woolley2021bluetooth}. We consider a convolutional code of rate \sfrac{1}{2} for the PHY layer and a variable rate convolutional code with a rate as low as \sfrac{1}{12} for the MAC layer. The standard hard Viterbi decoder is used for both PHY and MAC layers.

For completeness, we consider two additional settings. The first is the case of uncoded transmission for the PHY layer and variable rate convolutional codes for the MAC layer, with hard Viterbi decoding. Secondly, we consider a highly reliable FEC used in 5G NR: the standard rate \sfrac{3}{4} LDPC codes~\cite{3gppts38212} for the PHY layer and variable rate LDPC codes for the MAC layer, with the standard min-sum decoder using a maximum of 6 iterations.\vspace{.2em}

\textbf{Baselines.} We consider two baselines: First, we compare CE-HARQ with the conventional HARQ scheme. Additionally, to clearly quantify the gain obtained by combining FEC with compressed error retransmissions, we establish the following baseline scheme, by modifying CE-HARQ: In the MAC layer, the FEC is removed and the error vector is compressed after each retransmission. The PHY layer then encodes the compressed error vector at a constant coding rate. Owing to it's similarity with AIC, we refer to the scheme as AIC with Arithmetic Coding (AIC-AC). 



We evaluate the performance of CE-HARQ against the baselines on two practical scenarios: (1) Low-latency: Minimize BLER for a fixed number of retransmissions, and (2) High-reliability: Minimize the average number of rounds to achieve a target BLER. Furthermore, we evaluate and compare the spectral efficiency of our algorithm against the baselines. \vspace{.2em}

\begin{figure*}[!htbp]

\centerline{
\subfigure[]
{
    \centering
 	\includegraphics[width=0.35\linewidth]{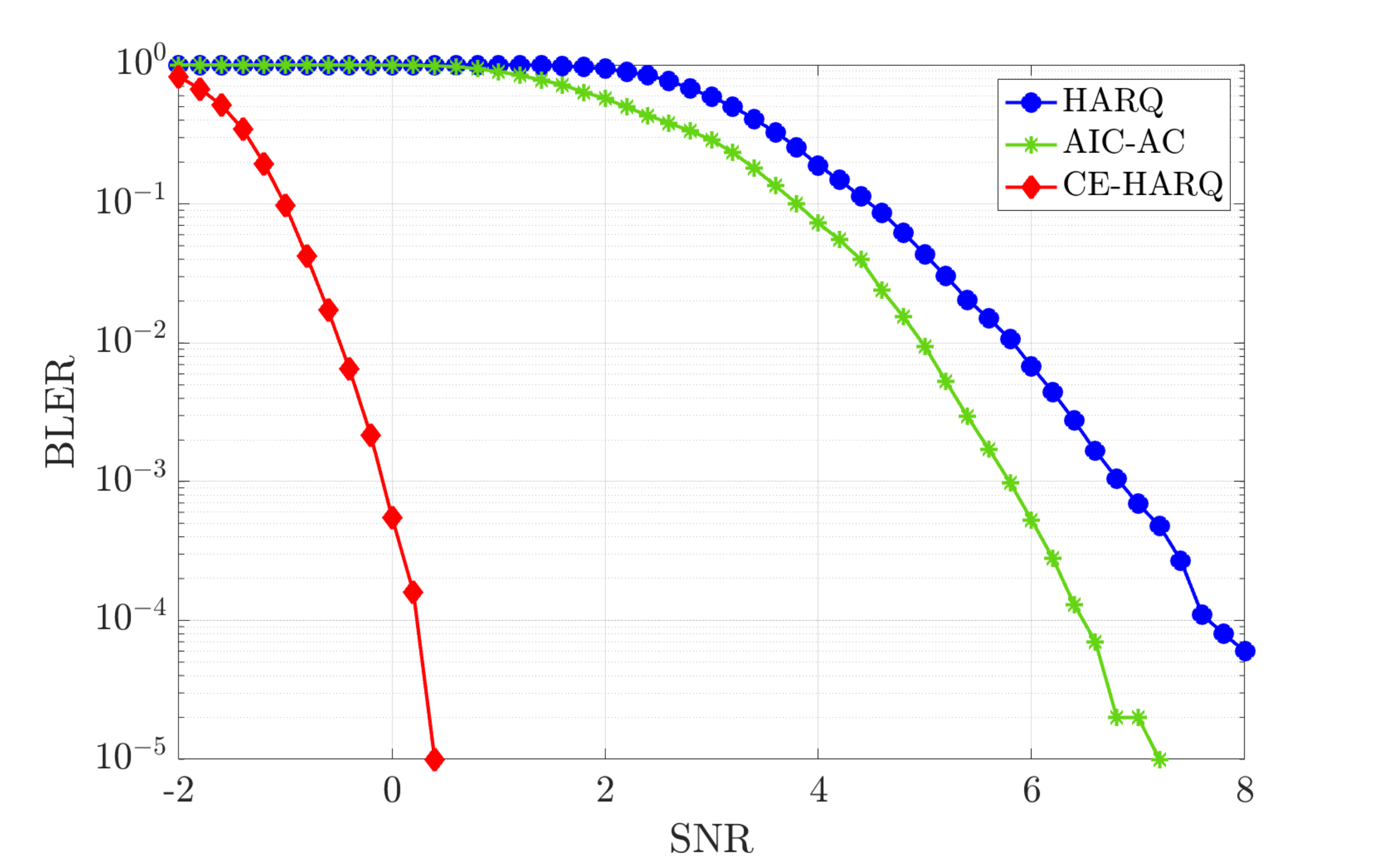}
 	\label{fig:bler_comp_uncoded}
}
\hspace{-0.25in}
\subfigure[]
{
    \centering
    \includegraphics[width=0.35\linewidth]{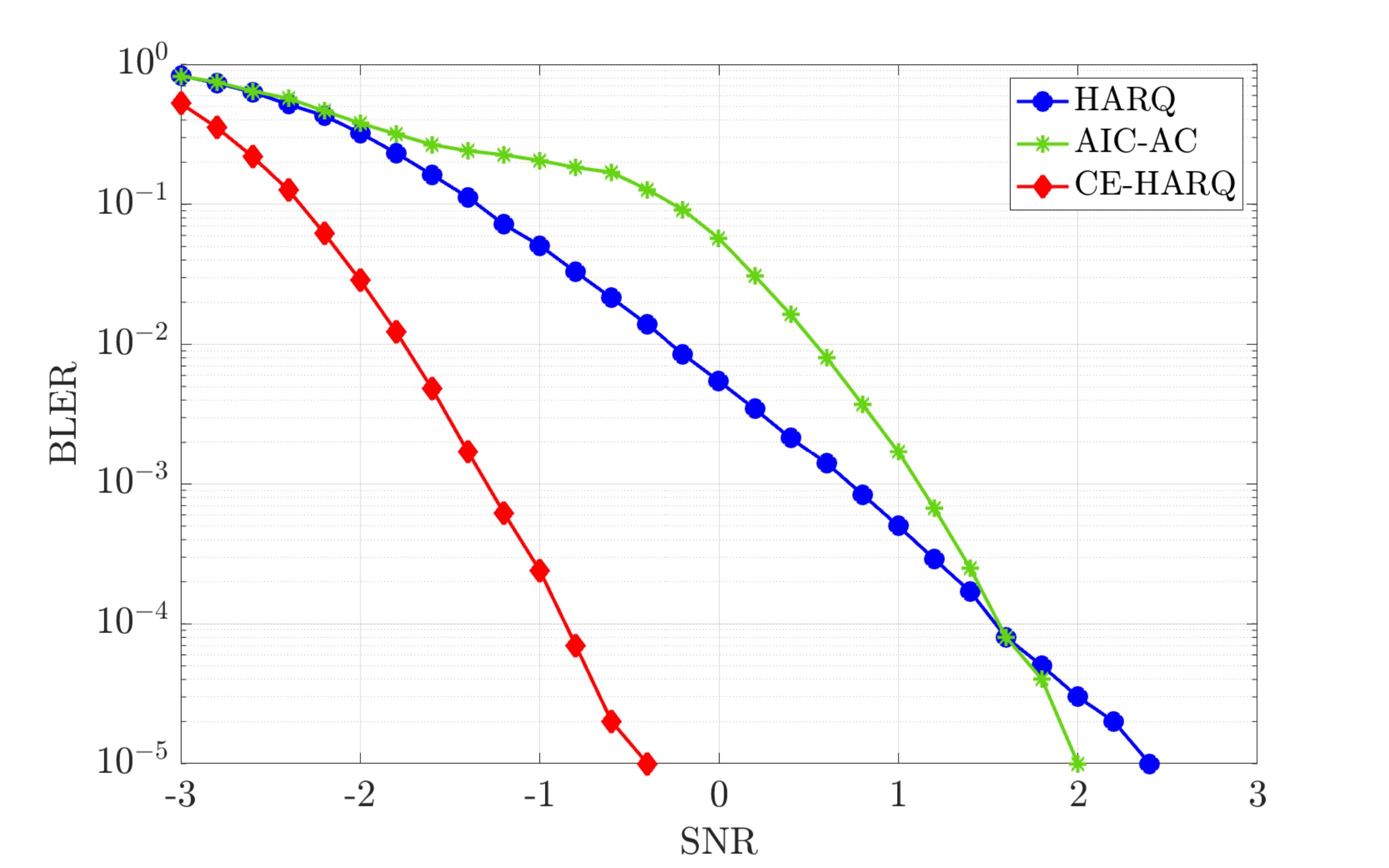}
 	\label{fig:bler_comp_conv}
  }
\hspace{-0.25in}
\subfigure[]
  {
    \centering
    \includegraphics[width=0.35\linewidth]{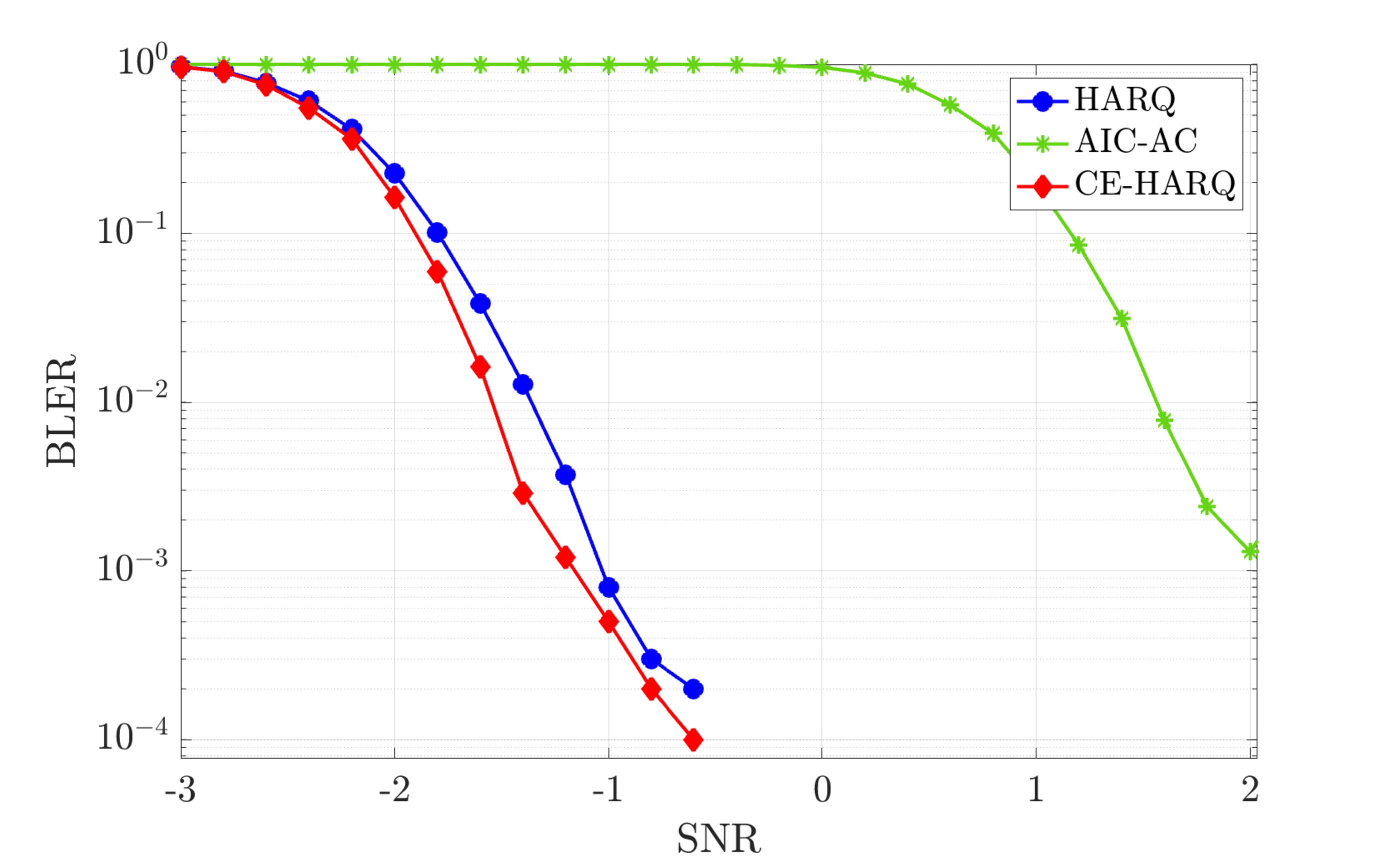}
 	\label{fig:bler_comp_ldpc}
  }

  }
  \caption{BLER performance after 3 retransmissions, when the PHY coding scheme is: a) Uncoded, message length $K=800$. b) Rate \sfrac{1}{2} Convolutional code, $K=800$. c) Rate \sfrac{3}{4} LDPC code, $K=960$. CE-HARQ exhibits large gains over the conventional HARQ and AIC with arithmetic coding, which always transmits the compressed error, with low-complexity PHY coding schemes such as convolutional codes, typically encountered in resource-constrained devices.}

\end{figure*}

\textbf{Performance for a given latency}. In low-latency applications, the maximum number of retransmissions is often limited. Here, we consider a maximum of 3 retransmissions and compare the BLER performance. As demonstrated in~\figref{fig:bler_comp_uncoded} for uncoded PHY, CE-HARQ is better than HARQ and AIC-AC by up to $7.5$ dB and $6.3$ dB respectively at a BLER of $10^{-4}$. This translates to a significant improvement in the reliability and range of the system while maintaining the same latency. Similarly, ~\figref{fig:bler_comp_conv} shows that when PHY uses a convolutional code of rate \sfrac{1}{2}, CE-HARQ performs better than HARQ and AIC-AC by up to $2.5$ dB. Finally, when a \sfrac{3}{4} LDPC code is used for the PHY layer (~\figref{fig:bler_comp_ldpc}), CE-HARQ obtains a marginal gain of $0.1$ dB. In summary, CE-HARQ provides the maximum gains for systems with low-complexity coding schemes, which are inevitable in resource-constrained devices such as Bluetooth-LE. Additionally, it is clearly seen that in practical retransmission systems, the uncoded transmission of compressed errors is highly undesirable. \vspace{.2em}


\begin{figure*}[htbp]

\centerline{
\subfigure[]
{
    \centering
 	\includegraphics[width=0.35\linewidth]{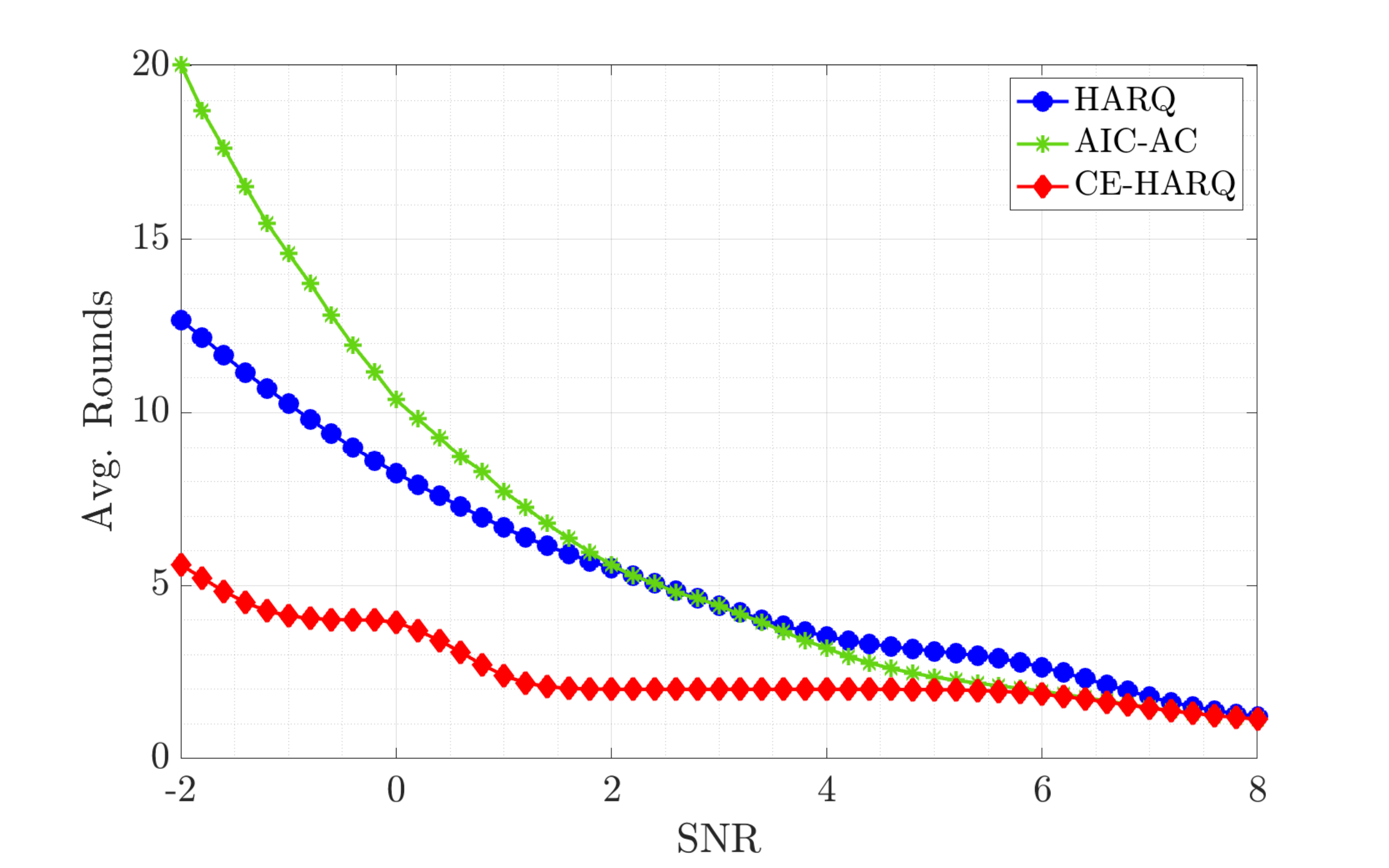}
 	\label{fig:latency_comp_uncoded}
}
\hspace{-0.25in}
\subfigure[]
{
    \centering
    \includegraphics[width=0.35\linewidth]{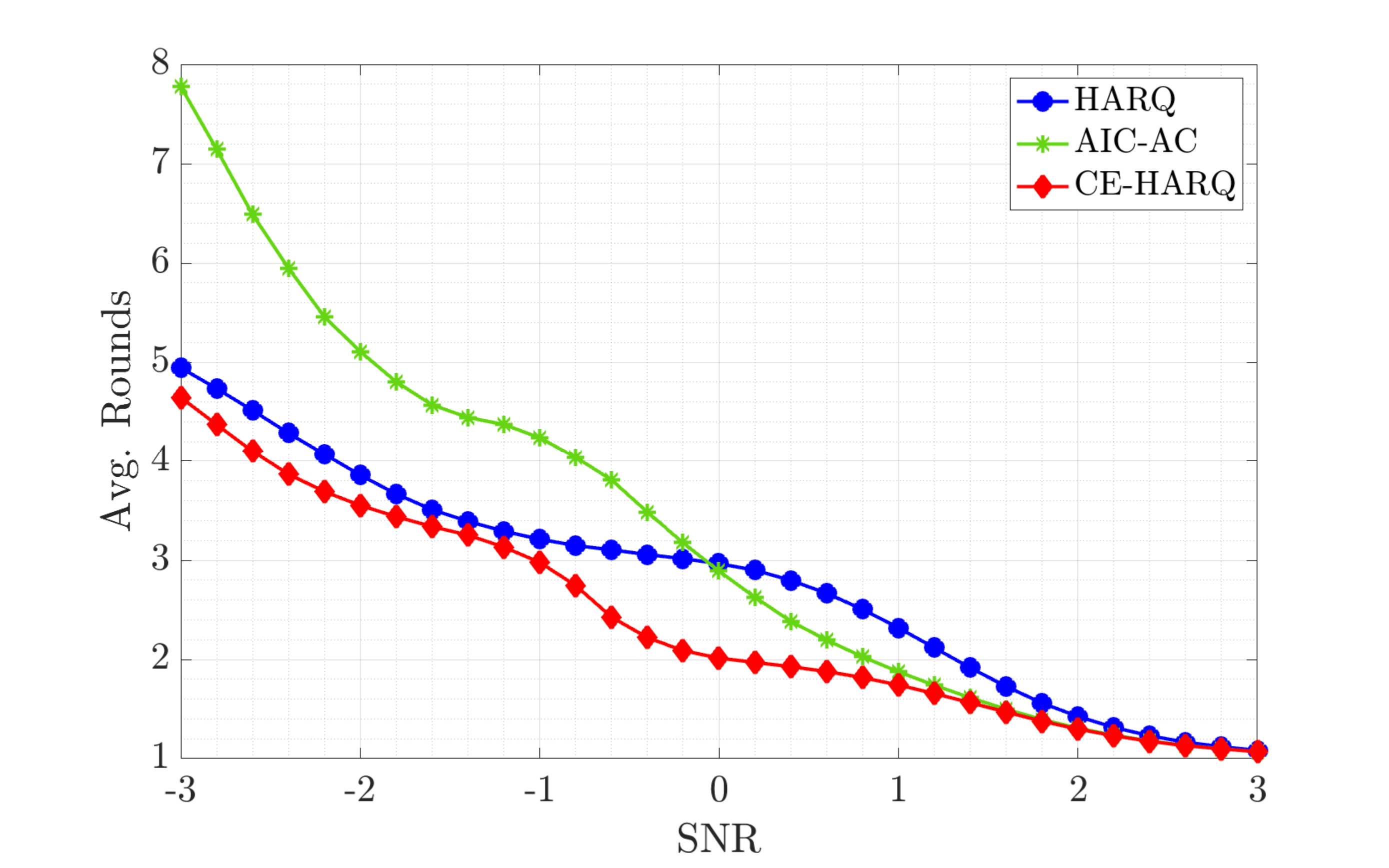}
 	\label{fig:latency_comp_conv}
  }
\hspace{-0.25in}
\subfigure[]
  {
    \centering
    \includegraphics[width=0.35\linewidth]{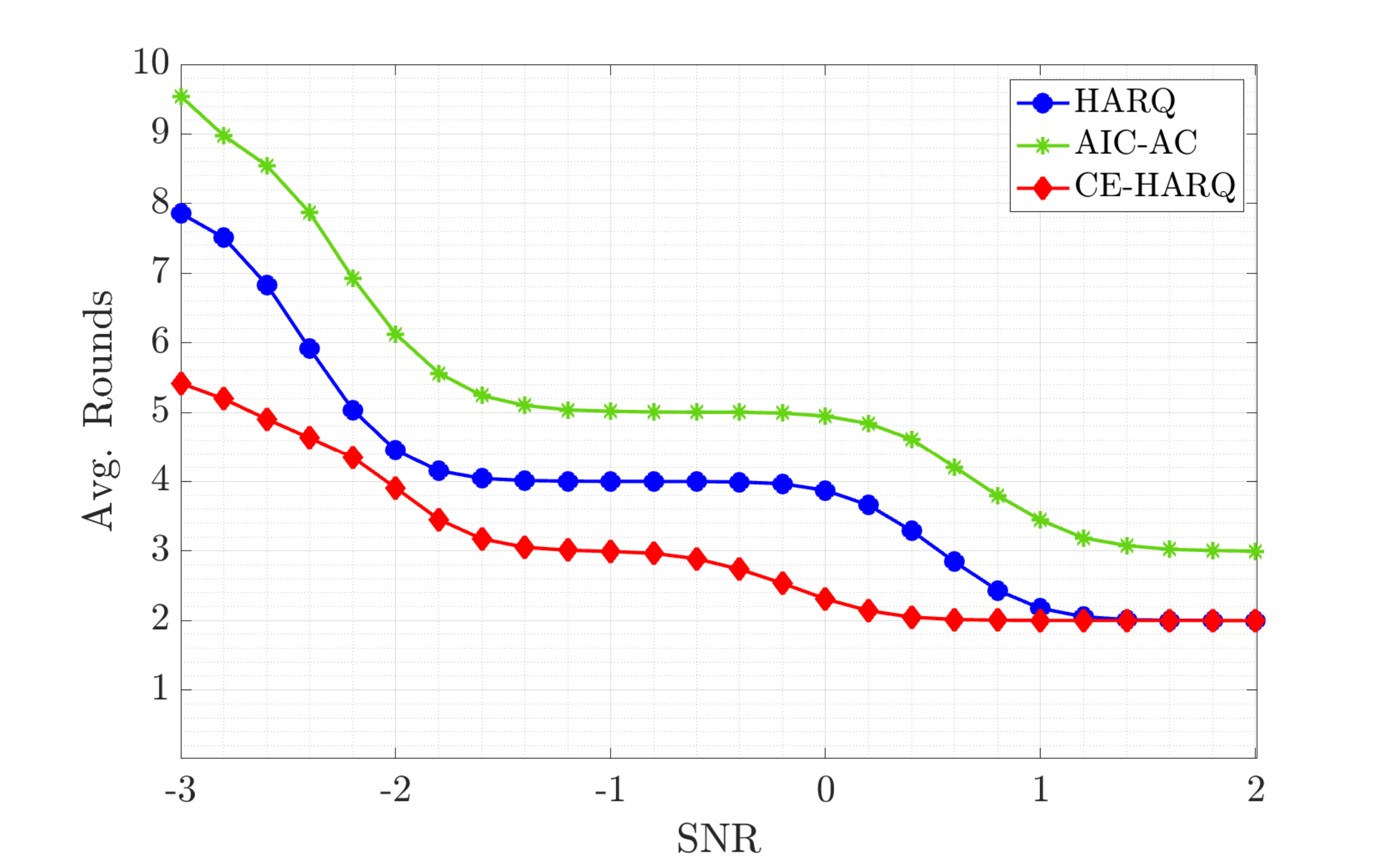}
 	\label{fig:latency_comp_ldpc}
  }

  }
  \caption{Average number of rounds for successful transmission, when the PHY coding scheme is: a) Uncoded, message length $K=800$. b) Rate \sfrac{1}{2} Convolutional code, $K=800$. c) Rate \sfrac{3}{4} LDPC code, $K=960$. CE-HARQ obtains maximum gain over baselines when the PHY layer is uncoded and the gains decrease as the complexity of FEC increases.}
\end{figure*}

\textbf{Latency for a given performance}. In high-reliability applications, it is critical to guarantee a low BLER. Here, for the above three settings, we now consider a target BLER of $10^{-5}$ and compare the average rounds required to achieve this. As highlighted in~\figref{fig:latency_comp_uncoded} for uncoded PHY, CE-HARQ has up to $50\%$ and $62\%$ lower latency compared to HARQ and AIC-AC respectively at SNR $0$ dB. This translates to significant improvements in power consumption while maintaining the same reliability. In~\figref{fig:latency_comp_conv}, \sfrac{1}{2} convolutional coding is used for the PHY layer resulting in latency improvements of up to $33\%$ for CE-HARQ compared to HARQ and AIC-AC at SNR $0$ dB. Finally, in~\figref{fig:latency_comp_ldpc}, \sfrac{3}{4} LDPC coding is used for the PHY layer, resulting in up to $42\%$ and $54\%$ lower latency compared to HARQ and AIC-AC respectively at SNR $0$ dB. \vspace{.2em}

\begin{figure*}[htbp]

\centerline{
\subfigure[]
{
    \centering
 	\includegraphics[width=0.35\linewidth]{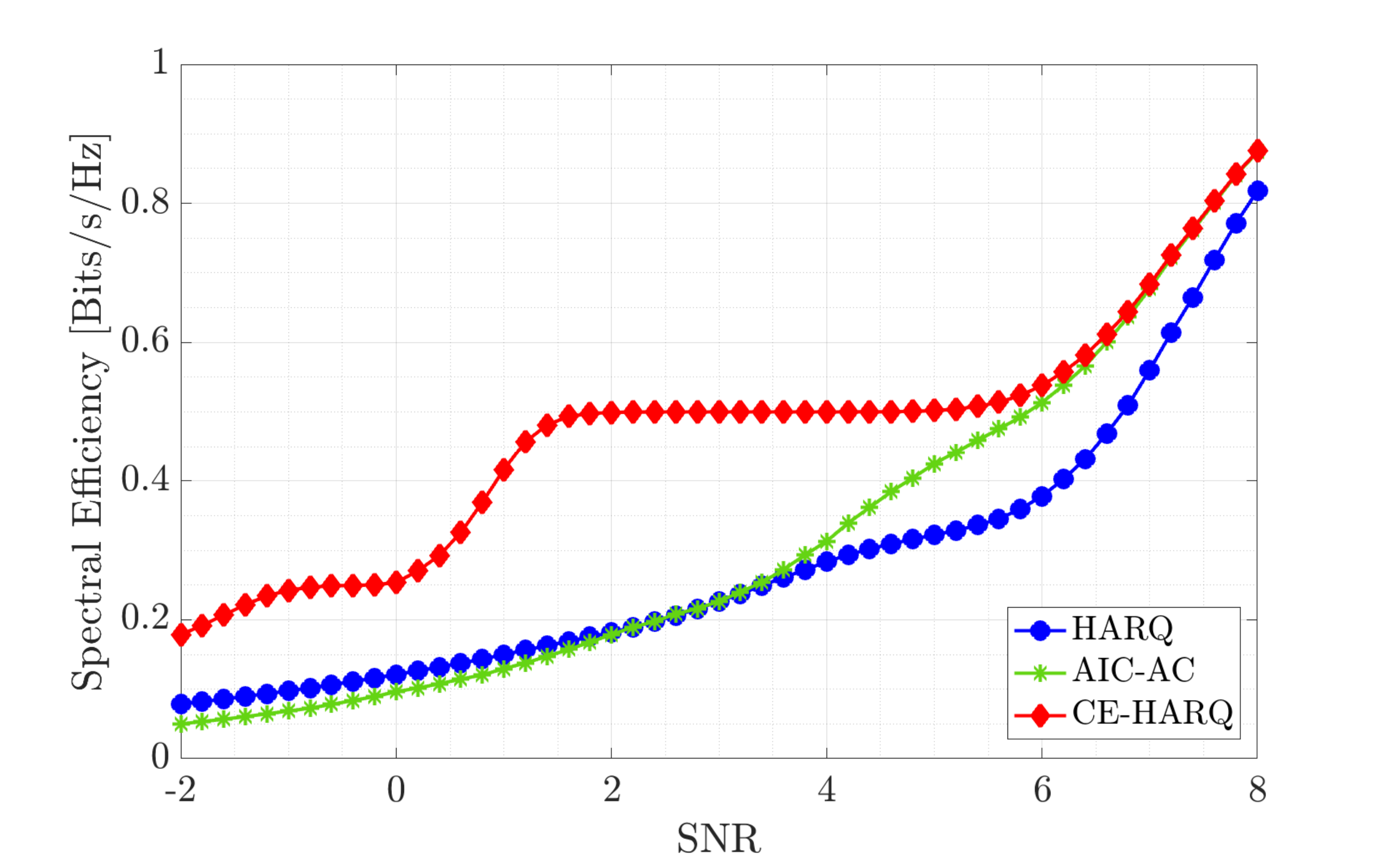}
 	\label{fig:se_comp_uncoded}
}
\hspace{-0.25in}
\subfigure[]
{
    \centering
    \includegraphics[width=0.35\linewidth]{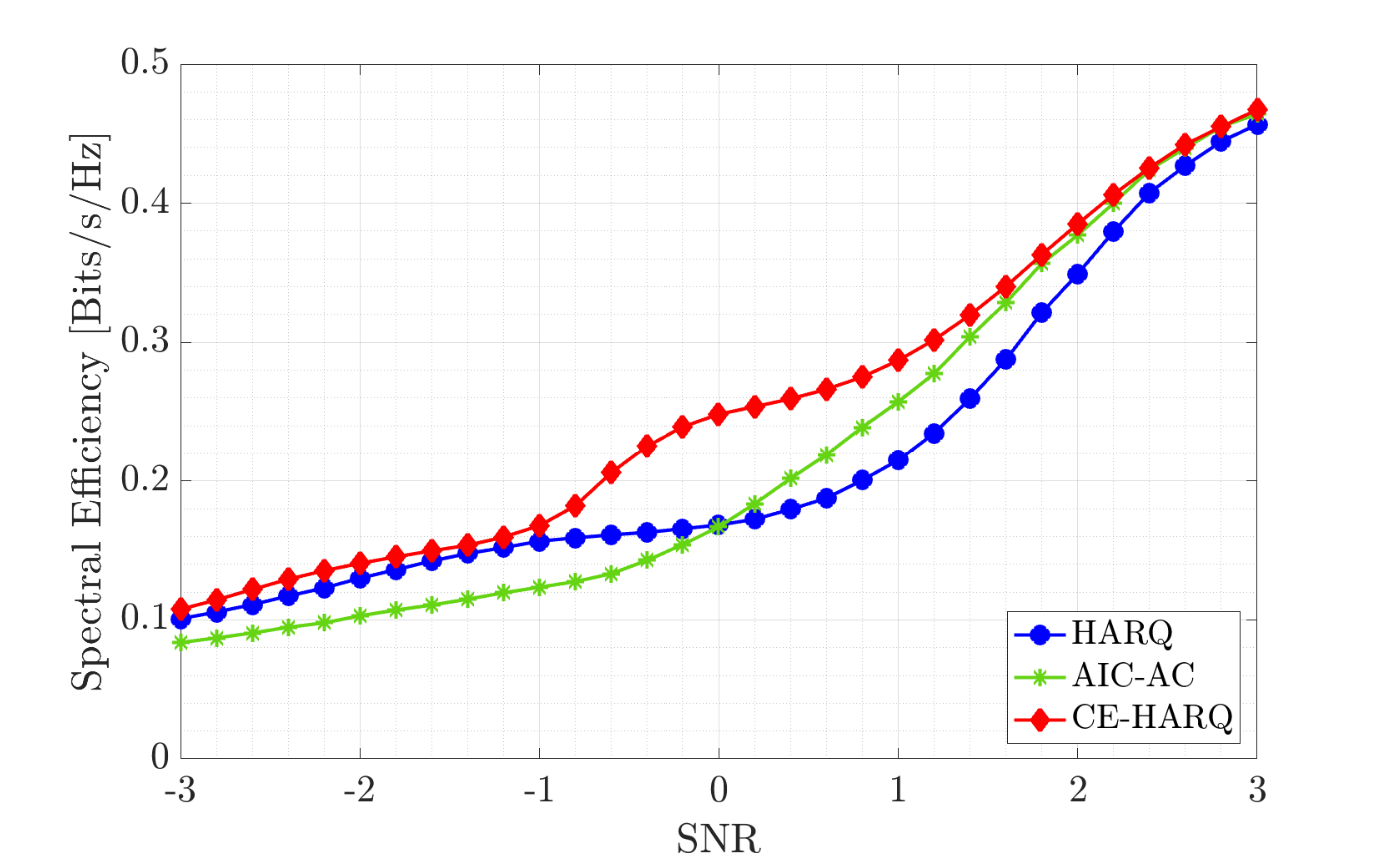}
 	\label{fig:se_comp_conv}
  }
\hspace{-0.25in}
\subfigure[]
  {
    \centering
    \includegraphics[width=0.35\linewidth]{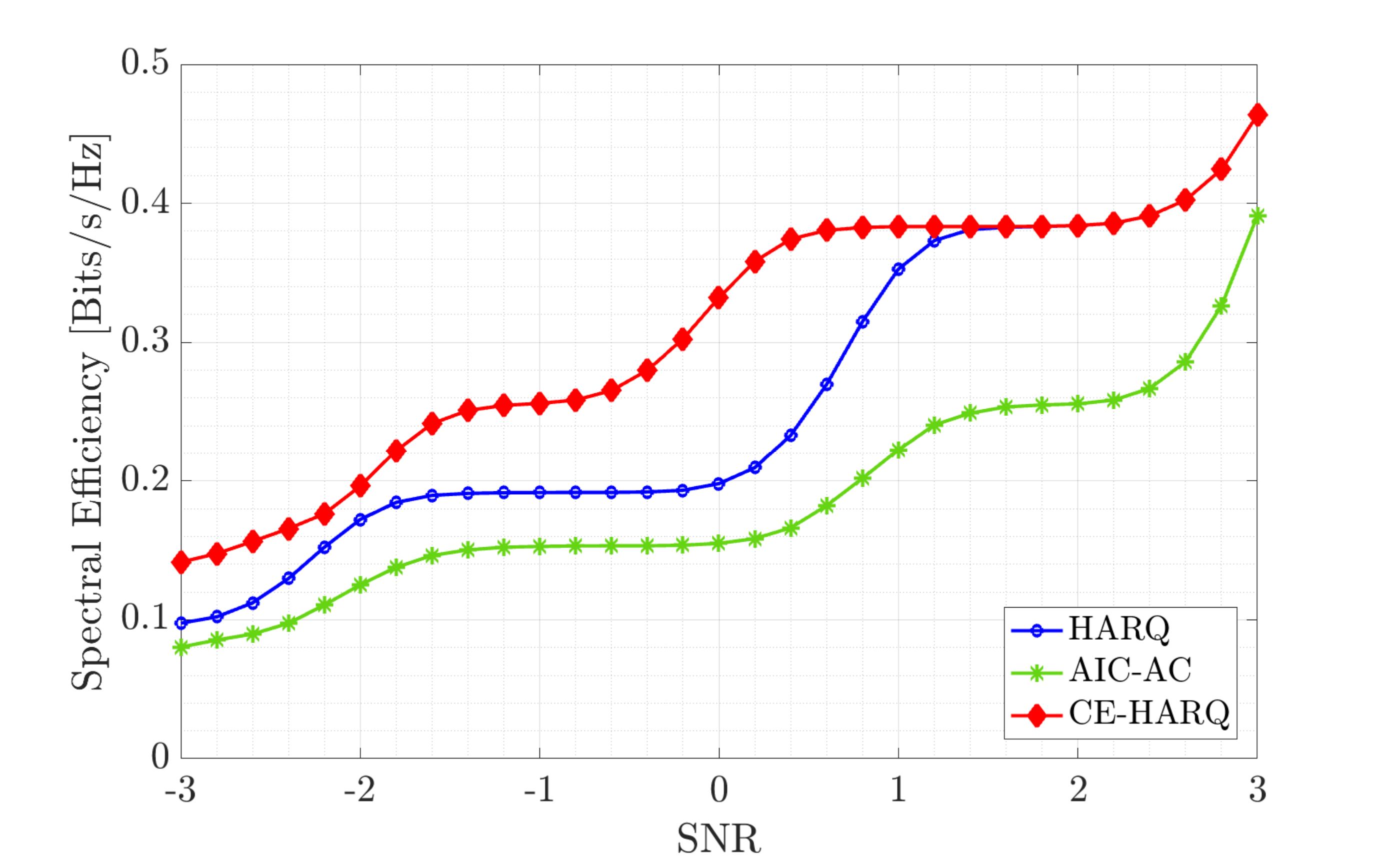}
 	\label{fig:se_comp_ldpc}
  }

  }
  \caption{Spectral Efficiency for a target BLER of $10^{-5}$, when the PHY coding scheme is : a) Uncoded, message length $K=800$. b) Rate \sfrac{1}{2} Convolutional code, $K=800$. c) Rate \sfrac{3}{4} LDPC code, $K=960$. CE-HARQ consistently achieves better spectral efficiency compared to the two baselines.}\label{fig:se_comp}

\end{figure*}

\textbf{Spectral efficiency. } The spectral efficiency of the forward channel for a given SNR and BLER is defined as 
\begin{align}
   \text{SE} \stackrel{\text{def}}{=} \frac{KQ}{E[N_c]}(1 - \text{BLER}) \text{bits/s/Hz},
\end{align}
where $K$ is the message length and $E[N_c]$ is the expected total channel uses for the retransmission scheme. $Q$ is the order of modulation, which is $1$ for BPSK. 

CE-HARQ is a constant-length code; hence $\E[N_c] = N\E[D_c]$, where $N$ is the codeword length and $\E[D_c]$ is the average number of rounds required to achieve the target BLER. 

SE combines the error correction and latency to provide a unified view of the trade-off between both metrics. While AIC-AC is optimized for SE by utilizing the minimal number of channel uses in every round, it is not straight forward how the trade-off between lesser number of rounds with $N$ channel uses and higher number of rounds with $H(\mathbf{e}^{(i)})$ channel uses plays out. ~\figref{fig:se_comp} shows that CE-HARQ consistently achieves better SE than HARQ and AIC-AC.

\vspace{.2em}
Through extensive simulations across different settings, we show that CE-HARQ significantly improves upon the baselines in resource-constrained setups such as Bluetooth, where using low-complexity coding schemes is unavoidable.  
However, the gains over HARQ are minimal when strong coding schemes like LDPC are used. This can be attributed to the narrow waterfall region in such codes, which results in the coding gain of the CE phase being inferior to the gains obtained by increased effective SNR through HARQ.

\section{Conclusion and Remarks}
In this work, we provide a practical framework, Compressed Error HARQ, to build block-feedback codes for noise-asymmetric channels. We assume a noise-asymmetric channel and focus on resource-constrained settings where the transmitter is a low-power, low-memory device, and the receiver can operate at much higher power. CE-HARQ generalizes the Hybrid-ARQ scheme for noise-asymmetric channels by combining it with compressed error, forward error correction, and a novel retransmission selection mechanism that intelligently switches between HARQ and Compressed Error (CE) based on the sparsity of the error vector. 

CE-HARQ achieves significant improvements in both error correction performance and latency for a variety of configurations. Additionally, CE-HARQ achieves better Spectral Efficiency compared to purely CE based schemes such as AIC-AC. Significant gains are observed in the case of convolutional coding for the PHY layer, which is a common configuration for low-power applications such as Bluetooth and IoT systems. We also outline the design choices essential to facilitate practical deployment such as limiting the feedback to be binary and adding a fallback mechanism to ensure continual improvement of the message estimate at the receiver. 

There are several interesting directions for extending CE-HARQ. The first is the case of partial feedback where the full message estimate is not available, and only a few bits of information can be sent. Deciding the optimal information for feedback is a non-trivial problem. Another interesting problem is the case of noisy feedback and non-AWGN channels, where the expected gains will depend heavily on the quality of the feedback link. Characterizing such a trade-off would be a very interesting future research direction. 




 \vspace{-.8em} 
 \medskip
 \small
 \bibliographystyle{IEEEtran}
 \bibliography{bibilography}

\clearpage
\normalsize
\begin{appendices}
\section{Selection of Sparsity threshold}\label{sec:spa_selection}
At the beginning of round $i$, the error vector should be compressed only if it is sufficiently sparse. To compute the optimal sparsity threshold, two approaches are proposed. The first is an analytical approach that compares the expected probability of error for HARQ vs CE-HARQ, for a given round and SNR. Next, a simple analytical approach based on Monte Carlo simulations is discussed.  

\subsection{Analytically computing the optimal sparsity threshold}\label{sec:spa_analytical} The performance of convolutional codes is well studied in literature \cite{viterbi1971convolutional, malkamaki1999evaluating} and we assume that the expected bit error rate can be computed for a given SNR $S$ and rate $R$ as $P_e(S,R)$. For many short length block-codes, this can be easily computed using Monte Carlo simulations. 

For ARQ, since every retransmission has equal probability of error, compressing the error vector and using low-rate FEC always improves the decoding probability. But in HARQ type retransmission, there is a trade-off between the coding gain achieved by error compression and FEC vs the SNR gain of HARQ. Assuming chase combining at the receiver, the effective SNR at the receiver at round $i$ can be computed as 
\begin{align}
    \text{S}_i = S + 10\log_{10} i,
\end{align}
where SNR is the SNR in first round. With each retransmission, chase combining results in an effective SNR gain via the principle of Maximal Ratio Combining \cite{frenger2001performance}. 

During round $i$, HARQ and CE-HARQ can achieve a maximum possible effective SNR of 
\begin{align}
    S_H &= S + 10\log_{10} D, \\
    S_i &= S + 10\log_{10} (D-i+1)
\end{align}
 respectively, by the end of maximum number of rounds $D$. Here, $S$ is the SNR in round 1. For a given error vector $e^{(i)}$ in round $i$, CE-HARQ can achieve the lowest FEC rate of $R_i = \sfrac{H(e^{(i)})}{N}$. HARQ uses a FEC rate $R = \sfrac{K}{N}$ for all the rounds. 

 Assuming independent and identically distributed (i.i.d) errors in the message vector, the entropy of the error vector $e^{(i)}$ can be written as
\begin{align}
    H(e^{(i)}) = -K(\tau_i \log_{2} \tau_i + (1-\tau_i)\log_{2} (1-\tau_i)),
\end{align}
where $\tau_i = \frac{|e^{(i)}|}{K}$ is sparsity of the error vector. 

Thus, for round $i$, we define optimal threshold $\tau_i^\star$ as the maximum value of sparsity $\tau_i$ such that $P_e(S_i, R_i) < P_e(S_H,R)$. 

\subsubsection{Grid search for threshold}\label{sec:grid_search}
Since it is not feasible to get a closed form solution for $\tau_i^\star$, we instead explore a simple grid search based approach. To keep the search computationally feasible, we fix single optimal threshold $\tau^\star$ for all the rounds of communication. For each SNR and coding rate, probability of decoding error is found via Monte Carlo simulations for different values of sparsity threshold from $0$ to $0.1$ in steps of $0.005$. Based on the results, the best sparsity threshold is chosen for each SNR.

\section{Ablation studies}
To highlight the significance of selecting the optimal sparsity threshold $\tau^\star$, we compare the performance of CE-HARQ with threshold $\tau^\star$ denoted by CE-HARQ($\tau^\star$) to a naive approach where the error is always compressed \textit{i.e.,} $\tau = 1$ followed by FEC, denoted by CE-HARQ($\tau$). \figref{fig:latency_comp_abl} demonstrates that the optimal selection of the sparsity threshold $\tau^\star$ crucially impacts the performance on both our evaluation settings. 

\begin{figure*}[htbp]

\centerline{
\subfigure[]
{
    \centering
 	\includegraphics[width=0.5\linewidth]{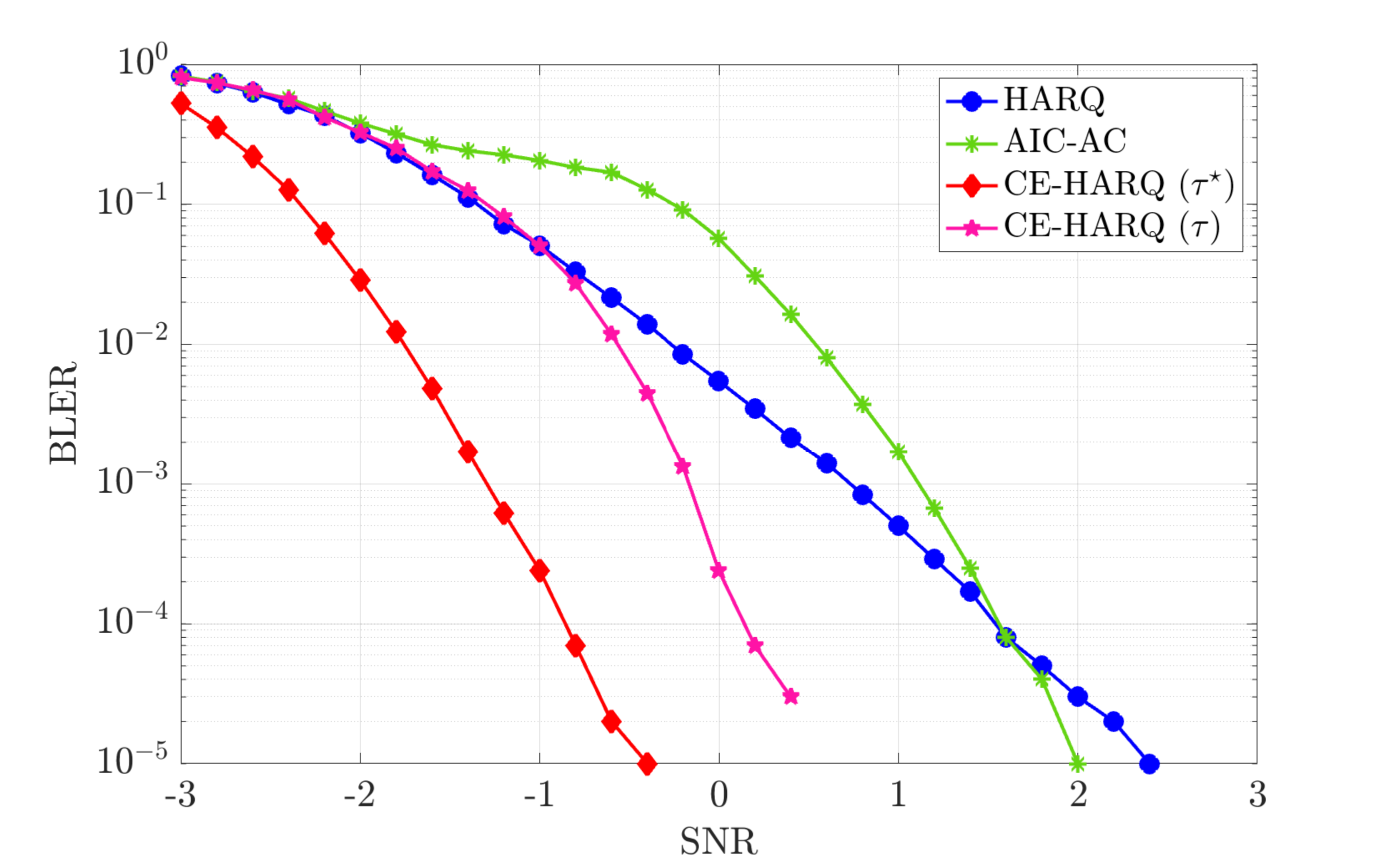}
 	\label{fig:latency_comp_uncoded_abl}
}
\subfigure[]
{
    \centering
    \includegraphics[width=0.5\linewidth]{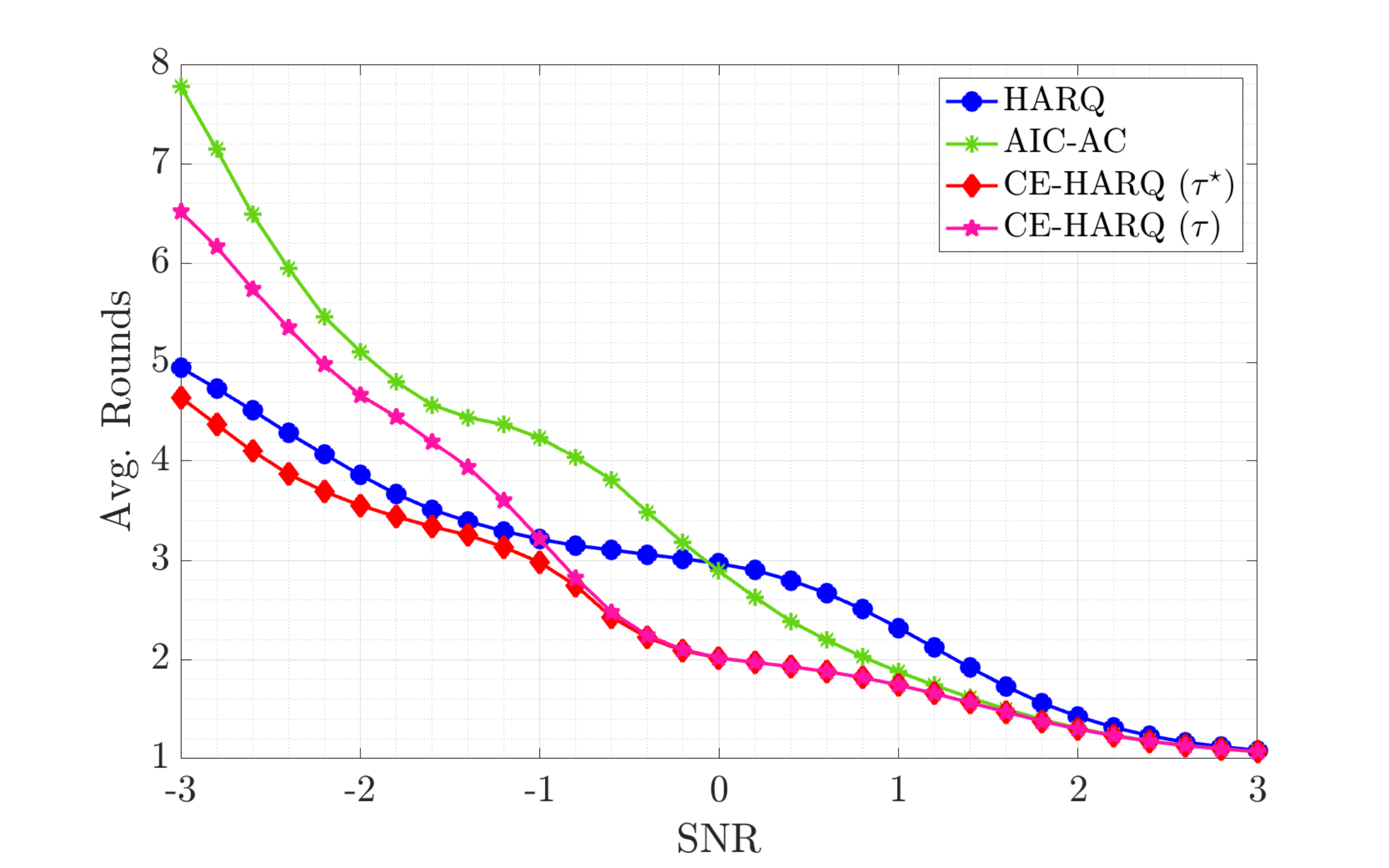}
 	\label{fig:latency_comp_conv_abl}
  }
  }
  \caption{Effect of choosing the sparsity threshold : Using a rate \sfrac{1}{2} convolutional code for PHY layer, CE-HARQ with optimal threshold $\tau^\star$ is significantly better than continuous compression ($\tau=1$) when : a) BLER is compared for maximum of 3 retransmissions b) Average rounds for error free transmission is compared. }
  \label{fig:latency_comp_abl}
\end{figure*}

\section{Source-Channel Separation for the transmission of error}\label{sec:src-ch separation}

Consider the following setup \figref{fig:source_channel}, where $\mathbf{U}$ denotes $K$ i.i.d. discrete sources available at the encoder, $\mathbf{S}$ denotes $K$ i.i.d. discrete sources that are correlated with $\mathbf{U}$ as $(U_i,S_i) \sim \text{i.i.d.\ } \Pr(u_i,s_i)$ for $i=1,\cdots,K$  
available at both the encoder and decoder as side information. The encoder maps $(\mathbf{U},\mathbf{S})$ to a length-N discrete sequence $\mathbf{X}$, which is corrupted by a noisy channel as $\mathbf{Y}$. The decoder maps $\mathbf{Y}$ to an estimate for $\mathbf{U}$, denoted by $\mathbf{V}$. This setup models the $i$-th round of communication from the transmitter to the receiver. The transmitter has a message $\mathbf{U} \in \{0,1\}^K$ and the receiver's estimated message as side information $\mathbf{S} = \mathbf{U} \oplus \mathbf{e}^{(i-1)}$; the receiver also has its own estimated message from the previous round of communication $\mathbf{S} = \mathbf{U} \oplus \mathbf{e}^{(i-1)}$ as side information.

\begin{figure}[ht]
    \centering
 	\includegraphics[width=0.6\linewidth]{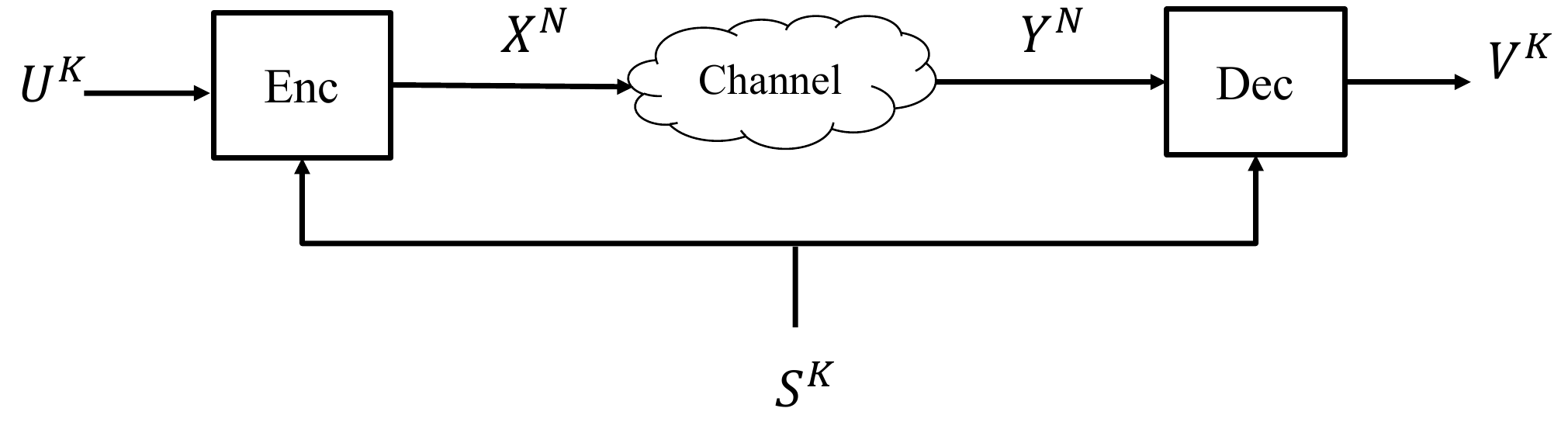}
 	\captionsetup{font=small}
 	\caption{The $i$-th round of forward communication can be viewed as joint source-channel coding with Tx/Rx side information, where $S^K = \hat{U}^{(i-1)}$ (previous message estimate) serves as Tx/Rx side info. Theorem 1 establishes the optimal rate-distortion, which is achieved by a separated source-channel coding.} 
 	\label{fig:source_channel}
 \end{figure}

\begin{theorem} 
Let $d(u,v)$ denote a distortion-metric. A rate-$K/N$ code achieves the distortion $D$ if and only if $K/N \leq C/R_\text{SI}(D)$, where $C = \max_{p(x)}I(X;Y)$ and 
\[
R_\text{SI}(D) = \min_{p(v|u,s): \E[d(u,v)] \le D} I(U;V|S). 
\]
\end{theorem}

\begin{proof}
\noindent Achievability follows immediately from the separate source-channel coding. Converse also closely follows the converse proof of the source-channel separation theorem for channels without side information. 

Suppose $\frac{1}{K}\E[d(\mathbf{U},\mathbf{V})] \le D$. Then, 
\begin{align}
    I(U^K;V^K|S^K) &=  \sum_{i=1}^K I(U_i; V^K|U^{i-1}, S^K) \nonumber\\
    &\stackrel{(a)}{=} \sum_{i=1}^K I(U_i; V^K,U^{i-1}|S^K) \nonumber\\
    &\ge  \sum_{i=1}^K I(U_i; V_i|S^K)\nonumber\\
    &\ge \sum_{i=1}^K I(U_i; V_i|S_i)\nonumber\\
    &\ge K R_\text{SI}(D) \label{proof-eq1}
\end{align}
where (a) holds since $U_i$ is a memoryless source. 

Now, by the data processing inequality, we have 
\begin{align}
\frac{1}{K} I(U^K;V^K|S^K) \le \frac{1}{K}I(U^K;Y^N|S^K).    \label{proof-eq2}
\end{align}
Formally, the inequality above holds since $U^K-(Y^N,S^K)-V^K$ forms a Markov Chain. 

Finally, from the converse for the communication over discrete memoryless channels, we have 
\begin{align}
I(U^K;Y^N|S^K) &= \sum_{i=1}^N I(U^K;Y_i|Y^{i-1},S^K)\nonumber\\
& \leq \sum_{i=1}^N I(U^K,Y^{i-1};Y_i|S^K)\nonumber\\
& \leq \sum_{i=1}^N I(X_i,U^K,Y^{i-1};Y_i|S^K)\nonumber\\
& \stackrel{(a)}{=} \sum_{i=1}^N I(X_i;Y_i|S^N)\nonumber\\
& \leq  \sum_{i=1}^N I(X_i;Y_i|S_i)\nonumber\\
&\leq  \sum_{i=1}^N I(X_i,S_i;Y_i)\nonumber\\
&\leq  \sum_{i=1}^N I(X_i;Y_i)\nonumber\\
&\leq N C, \label{proof-eq3}
\end{align}
where (a) holds since $(U^K,Y^{i-1}) - X_i - Y_i$ forms a Markov chain given $S^K$. By combining equations~\eqref{proof-eq1},\eqref{proof-eq2}, and \eqref{proof-eq3}, we complete the proof. 
\end{proof}

\end{appendices}
\end{document}